\documentclass[11pt]{article}
\usepackage[authoryear,round]{natbib}
\usepackage{float}
\usepackage{graphicx,amsmath,amssymb,mathrsfs,amsfonts,amsthm,color,slashbox,amsbsy,ragged2e}
\usepackage{hyperref}
\usepackage{subfigure}
\usepackage{epsfig}
 \usepackage{lineno}

\newcommand{\Gb}{{\bar{G}}}
\newcommand{\Fb}{{\bar{F}}}
\newcommand{\xb}{{\boldsymbol{ x}}}

\newcommand{\al}{\alpha}
\newcommand{\be}{\beta}
\newcommand{\la}{\lambda}

\newcommand{\lab}{\boldsymbol{\la}}
\newcommand{\pb}{\boldsymbol{p}}
\newcommand{\ub}{\boldsymbol{u}}

\newcommand{\bq}{\begin{equation}}
\newcommand{\eq}{\end{equation}}
\newcommand{\bqs}{\begin{equation*}}
\newcommand{\eqs}{\end{equation*}}
\newcommand{\bqa}{\begin{eqnarray}}
\newcommand{\eqa}{\end{eqnarray}}
\newcommand{\bqas}{\begin{eqnarray*}}
\newcommand{\eqas}{\end{eqnarray*}}
\newcommand{\bc}{\begin{cases}}
\newcommand{\ec}{\end{cases}}
\newcommand{\bt}{\begin{thm}}
\newcommand{\et}{\end{thm}}

\newtheorem{theorem}{Theorem}[section]
\newtheorem{lemma}{Lemma}[section]
\newtheorem{corollary}{Corollary}[section]
\newtheorem{example}{Example}[section]
\newtheorem{definition}{Definition}[section]

\title{Ordering the smallest claim amounts from two sets of interdependent heterogeneous portfolios}
\author{ Hossein Nadeb, Hamzeh Torabi, Ali Dolati\\
Department  of Statistics, Yazd University,  Yazd, Iran,\\
}

\begin{document}
\date{}
\maketitle
\begin{abstract}
Let $ X_{\la_1},\ldots,X_{\la_n}$ be a set of dependent and non-negative
random variables share a survival copula and let
$Y_i= I_{p_i}X_{\la_i}$, $i=1,\ldots,n$, where $I_{p_1},\ldots,I_{p_n}$ be independent Bernoulli random variables independent of
$X_{\la_i}$'s, with ${\rm E}[I_{p_i}]=p_i$, $i=1,\ldots,n$. In
actuarial sciences, $Y_i$ corresponds to the claim amount in a
portfolio of risks. This paper considers comparing the smallest claim amounts from two sets of
interdependent portfolios, in the
sense of usual and likelihood ratio orders, when the variables in one set have the parameters
$\lambda_1,\ldots,\lambda_n$ and $p_1,\ldots,p_n$ and the variables in the other set  have
the parameters $\lambda^{*}_1,\ldots,\lambda^{*}_n$ and $p^*_1,\ldots,p^*_n$. Also, we present some bounds for survival function of the smallest claim amount in a portfolio. To illustrate validity of the results, we serve some applicable models.

\end{abstract}
%{\bf MSC 2010:}  62F03, 62F10. 
%% or \MSC[2008] code \sep code (2000 is the default)

%%
%% Start line numbering here if you want
%%
%\linenumbers

%% main text
\section{Introduction}
Suppose that $X_{\la_1},\ldots,X_{\la_n}$, assuming $X_{\la_i}$ has the survival function $\Fb(x;\la_i)$, are non-negative random variables denoting the total random severities of $n$ policyholders in an insurance period. Further, let $I_{p_1},\ldots,I_{p_n}$ be a set of independent
Bernoulli random variables, with $I_{p_i}$ is corresponding to $X_{\la_i}$, such that
$I_{p_i}=1$ whenever the $i$th policyholder makes random claim amount
$X_{\la_i}$ and $I_{p_i}=0$ whenever does not make a claim. In this
notation, $Y_i= I_{p_i} X_{\la_i}$ is the claim amount related to $i$th policyholder and $(Y_1,\ldots,Y_n)$ is said to be a portfolio of
risks. Further, consider another portfolio of risks $(Y^*_1,\ldots,Y^*_n)$ with the parameter vectors $(\la^*_1,\ldots,\la^*_n)$ and $(p^*_1,\ldots,p^*_n)$.

The annual premium is the amount received by the insurer which is the primary cost to accept the risk. Determining the
annual premium is very important problem for the insurance companies. Therefore, deriving preferences
between random future gains or losses is an appealing topic for the actuaries. For this purpose, stochastic orders are very
helpful. Stochastic orders have been extensively used in the areas management science, financial economics,
insurance, actuarial science, operation research, reliability
theory, queuing theory and survival analysis. For a comprehensive discussions on
stochastic orders, one may refer to M\"{u}ller and Stoyan\cite{must}, Shaked and Shanthikumar\cite{ss} and Li and Li\cite{lll}.

The problem of orderings of some statistics in the portfolios $(Y_1,\ldots,Y_n)$ and $(Y^*_1,\ldots,Y^*_n)$, such as
the number of claims, $\sum_{i=1}^n I_{p_i}$, the aggregate claim amounts, $\sum_{i=1}^n Y_i$, the smallest, $Y_{1:n}=\min(Y_1,\ldots,Y_n)$,
and the largest claim amounts, $Y_{n:n}=\max(Y_1,\ldots,Y_n)$, have been
discussed in many researches; see, e.g., Karlin and
Novikoff\cite{kar}, Ma\cite{ma}, Frostig\cite{fro}, Hu and Ruan\cite{huru}, Denuit and Frostig\cite{defr}, Khaledi and Ahmadi\cite{khah}, Zhang and Zhao\cite{zz}, Barmalzan et al.\cite{bar1}, Li and Li\cite{lili}, Barmalzan et al.\cite{bar2018}, Barmalzan and Najafabadi\cite{bana}, Barmalzan et al.\cite{bar3}, Barmalzan et al.\cite{bar2}, Balakrishnan et al.\cite{baet} and Li and Li\cite{lili2}.

The most of published articles consider the case that the severities are independent, while sometimes this assumption is not satisfied and many of policies are simultaneously at risk, such as when earthquakes or epidemics occur. Here, the severities have a positive dependence.

In this paper, it is assumed that $X_{\la_1},\ldots,X_{\la_n}$ are non-negative and continuous random variables with the joint survival
function $\bar{H}(x_1,\ldots,x_n)$, marginal survival functions $\Fb(x;\la_1),\ldots,\Fb(x;\la_n)$, and the survival copula $C_R$ through the relation $\bar{H}(x_1,\ldots,x_n)=C_R(\Fb(x_1;\la_1),\ldots,\Fb(x_n;\la_n))$ in the view of the Sklar's Theorem; see Nelsen\cite{nel}. Here, we compare the smallest claim amounts arising from two sets of interdependent heterogeneous portfolios and then mainly focus on presenting some bounds for the survival function of the smallest claim amount in a set of interdependent heterogeneous portfolio.

The rest of the paper is organized as follows. In Section \ref{sec2}, we recall some definitions
and lemmas which will be used in the sequel. Subsection \ref{sub1} provides orderings of the smallest
claim amounts from two interdependent heterogeneous portfolios of risks for a general model in the sense
of the usual stochastic order. Also, it considers the proportional hazard rate model and provides some characterizations on the likelihood ratio order of the smallest claim amounts under some certain conditions. Subsection \ref{sub2} presents some useful lower and upper bounds for the survival function of the smallest claim amount and it establishes some numerical examples to illustrate the validity of the shown results.

%%%%%%%%%%%%%%%%%%%%%%%%%%%%%%
\section{The basic definitions and some prerequisites}\label{sec2}
In this section, we state some notions of stochastic orders,
majorization, weak majorization and  some lemmas which are needed to prove our main results.
Throughout the paper, we use the notations $ \Bbb R =(-\infty,+\infty) $, $ \Bbb R_+ =[0,+\infty) $ and $\bar{x}=\frac{1}{n}\sum_{i=1}^n x_i$. Also, we use the notion of increasingness for a function $g:\mathscr{A}\rightarrow {\Bbb R}$, $ \mathscr{A}\subseteq {\Bbb R}^{n} $, if it is non-decreasing in each argument. Similarly,  the notion of decreasingness is used, when $g$ is non-increasing in each argument.

Let $  X $ and $ Y $ be two non-negative random variables with the distribution functions $ F $ and $ G $, the survival functions $ \bar{F}=1 - F  $ and $ \bar{G}=1 - G  $, the density functions $f$ and $g$ and the hazard rate functions $ r_{X}={f}/{\bar{F}} $ and $ r_{Y}={g}/{\bar{G}} $, respectively.

\begin{definition}
$ X $ is said to be smaller than $ Y $ in the 
\begin{itemize}
\item[\rm (i)] usual stochastic order, denoted by $ X \leq_{\rm st} Y $, if $ \bar{F}(x)\leq\bar{G}(x) $ for all  $ x \in \Bbb R $;
\item[{\rm(ii)}] hazard rate order, denoted by $ X\leq_{\rm hr}Y $, if $ r_{Y} (x) \leq r_{X}(x) $ for all $ x \in \Bbb R $;
\item[{\rm(iii)}] likelihood ratio order, denoted by $ X\leq_{\rm lr}Y $, if $\frac{g(x)}{f(x)}$ is increasing in $ x \in \mathbb{R}_+ $.
\end{itemize}
\end{definition}

For a comprehensive discussion on various stochastic orders, we
refer to M\"{u}ller and Stoyan\cite{must}, Li and Li\cite{lll} and Shaked and Shanthikumar\cite{ss}.

The concepts of majorization of vectors and Schur-convexity and Schur-concavity of functions are also needed. For a
comprehensive discussion of these topics we refer to Marshall et
al. \cite{met}. We use the notation $ x_{1:n}\leq x_{2:n}\leq
\ldots \leq x_{n:n}$ to denote the increasing arrangement of components of the
vector $ \boldsymbol{x} = (x_{1}, \ldots ,x_{n})$.
\begin{definition}
 The vector $ \boldsymbol{x} $ is said to be
\begin{itemize}
\item[ (i)] weakly submajorized by the vector $ \boldsymbol{y} $ (denoted by $ \boldsymbol{x}\preceq_{\rm w}\boldsymbol{y} $) if
$\sum_{i=j}^{n}x_{i:n}\leq \sum_{i=j}^{n}y_{i:n}$ for all $j = 1, \ldots , n $,

\item[ (ii)] weakly supermajorized by the vector $ \boldsymbol{y} $ (denoted by $ \boldsymbol{x}\mathop \preceq \limits^{{\mathop{\rm w}} }\boldsymbol{y} $) if $ \sum_{i=1}^{j}x_{i:n}\geq \sum_{i=1}^{j}y_{i:n} $ for all $ j = 1, \ldots , n $,

\item[ (iii)] majorized by the vector $ \boldsymbol{y} $ (denoted by $ \boldsymbol{x}\mathop \preceq \limits^{{\mathop{\rm m}} }\boldsymbol{y} $) if $ \sum_{i=1}^{n}x_{i}= \sum_{i=1}^{n}y_{i}$ and $\sum_{i=1}^{j}x_{i:n}\geq \sum_{i=1}^{j}y_{i:n}$ for all $j = 1, \ldots , n-1 $.
\end{itemize}
\end{definition}
\begin{definition}
A real valued function $ \phi $ defined on a set $ \mathscr{A}\subseteq {\Bbb R}^{n} $ is said to be Schur-convex (Schur-concave) on $ \mathscr{A} $ if

\[
\boldsymbol{x} \mathop \preceq \limits^{{\mathop{\rm m}} }\boldsymbol{y} \quad \text{on}\quad \mathscr{A} \Longrightarrow \phi(\boldsymbol{x})\leq (\geq)\phi(\boldsymbol{y}).
\]
\end{definition}

\begin{lemma}[Marshall et al.\cite{met}, Theorem 3.A.8]\label{mkl}
A real valued function $\phi$ defined on a set $\mathscr{A}\subseteq \mathbb{R}^n$ satisfies
\begin{equation*}
\phi(\xb)\leq \phi(\boldsymbol{y})\quad \text{whenever}~~  \boldsymbol{x}\mathop \preceq \limits^{{\mathop{\rm w}} }\boldsymbol{y}~~  \text{on}~~ \mathscr{A},
\end{equation*}
if, and only if, $\phi$ is decreasing and Schur-convex on $\mathscr{A}$.
\end{lemma}

\begin{lemma}[Marshall et al.\cite{met}, 3.B.2]\label{mkl2}
Let $\phi:\mathbb{R}^n\rightarrow\mathbb{R}$ be a decreasing and Schur-convex function and $g:\mathbb{R}\rightarrow\mathbb{R}$ be an increasing and concave function. Then, the function $\psi(\xb)=\phi(g(x_1),\ldots,g(x_n))$ is decreasing and Schur-convex.
\end{lemma}

One of the needed concepts in this paper is the Archimedean copula. The class of Archimedean copulas having a wide range of dependence structures including the independent copula. In the following, we state some useful definitions and lemmas related to copulas.

\begin{definition}
A copula $C$ is called Archimedean if it is of the form $C(u_1,\ldots,u_n)=\phi^{-1}\left(\sum \limits_{i=1}^n \phi(u_i)\right)$, for $(u_1,\ldots,u_n)\in [0,1]^n$, which $\phi:[0,1]\rightarrow[0,\infty]$ is a strictly decreasing function, $\phi(0)=\infty$, $\phi(1)=0$ and $(-1)^k \frac{{\rm d}^k \phi (x)}{{\rm d}x^k}\geq 0 $, for $k\geq 0$, where $\phi^{-1}$ is the inverse of the function $\phi$. The function $\phi$ is called generator of the copula $C$.
\end{definition}

We state the following lemma from Durante\cite{dur} and Dolati and Dehghan Nezhad\cite{dd} related to Schur-concavity of Archimedean copulas. 

\begin{lemma}\label{l7}
Every Archimedean copula is Schur-concave.
\end{lemma}

\begin{definition}\label{def2}
A survival copula $C_R$ is positively upper orthant dependent ({\rm PUOD}), if for all $\ub \in [0,1]^n$, $C_R(\ub)\geq \prod_{i=1}^n u_i$.
\end{definition}

\begin{definition}\label{def1}
Let $C_R$ and $D_R$ be two survival copulas. $C_R$ is less {\rm PUOD} than $D_R$, denoted by $C_R\prec D_R$, if for all $\ub \in [0,1]^n$, $C_R(\ub)\leq D_R(\ub)$.
\end{definition}
%Note that the concept of PQD, first introduced by Lehmann\cite{leh} in the another form. 
\begin{definition}
Let $C$ be a copula. The main diagonal section of $C$ is the function $\delta_C:[0,1]\rightarrow [0,1]$ defined by $\delta_C(u)=C(u,\ldots,u)$.
\end{definition}

\begin{lemma}\label{l6}
For any copula $C$ and for all $\ub \in [0,1]^n$,
\begin{equation*}
\max\left(\sum_{i=1}^n u_i-n+1,0\right)\leq C(\ub)\leq \min \left(u_1,\ldots,u_n\right),
\end{equation*}
which, the bounds are called the Fr\'echet-Hoeffding bounds.
\end{lemma}

%An important copula in application, is the Farlie-Gumbel-Morgenstern (FGM) copula which introduced by Morgenstern\cite{morgen} with a trace back to Eyraud\cite{eyr} and discussed by Gumbel\cite{gum} and Farlie\cite{far}, of the form $C_{\theta}(\ub)=\prod_{i=1}^n u_i+\theta \prod_{i=1}^n u_i(1-u_i)$, where $\theta\in [-1,1]$.

%\begin{lemma}\label{l8}
%The FGM copula is Schur-concave for any $\theta \in [-1,1]$.
%\end{lemma}
For a comprehensive discussion in the topic of copula and the different types of dependency, one may refer to Nelsen\cite{nel}.

\section{Main results}\label{sec3}
This section consists of two subsections. In Subsection \ref{sub1}, we compare the smallest
claim amounts from two interdependent heterogeneous portfolios of risks in the sense
of the usual and the likelihood ratio orders. In subsection \ref{sub2}, some bounds for the survival function of the smallest claim amount are presented and some examples are established to illustrate the validity of the results.

\subsection{Stochastic comparison of the smallest claim amounts}\label{sub1}

The following theorem provides the usual stochastic order between the smallest
claim amounts in two heterogeneous portfolios of risks with the common parameter vectors $\lab$ and $\pb$ and different associated copulas.
\begin{theorem}\label{t1}
Let $ X_{\la_1},\ldots,X_{\la_n} $ be non-negative random variables with
$ X_{\la_i} \thicksim \Fb(x;\la_i)$, $i = 1, \ldots,n $, and the associated copula $C_R$. Further, suppose that $I_{p_1},\ldots,
I_{p_n}$ is a set of independent Bernoulli random variables, independent of the $X_{\la_i}$'s, with ${\rm E}[I_{p_i}]=p_i$, $i=1,\ldots,n$. Let $Y_i=I_{p_i} X_{\la_i}$, $i=1,\ldots,n$. Then we have
\begin{eqnarray*}
C^*_R\prec C_R \Longrightarrow Y^*_{1:n} \leq_{{\rm st}}Y_{1:n},
\end{eqnarray*}
where $Y_{1:n}$ and $Y^*_{1:n}$ are the smallest order statistics of $(Y_1,\ldots,Y_n)$ under the copula structures $C_R$ and $C^*_R$, respectively.
\end{theorem}
\begin{proof}
First, for any $x\geq 0$, the survival function of $Y_{1:n}$ is obtained as follows:
\begin{eqnarray}\label{dist1}
\Gb_{Y_{1:n}}(x)&=&{\rm P}(Y_{1:n}>x)\nonumber\\
&=&{\rm P}(I_{p_i} X_{\la_i}>x,~ \forall ~1\leq i \leq n)\nonumber\\
&=&{\rm P}(I_{p_i}=1,~\forall ~1\leq i \leq n){\rm P}\left(I_{p_i} X_{\la_i}>x,~ \forall ~1\leq i \leq n~|~I_{p_i}=1,~\forall ~1\leq i \leq n\right)\nonumber\\
&=&\left(\prod\limits_{i=1}^n p_i\right)~ {\rm P}\left(X_{\la_i}>x,~ \forall ~1\leq i \leq n\right)\nonumber\\
&=&\left(\prod\limits_{i=1}^n p_i\right)~ C_R\left(\Fb(x;\la_1),\ldots,\Fb(x;\la_n)\right).
\end{eqnarray}
Similarly, the survival function of $Y^*_{1:n}$ is given by
\begin{equation*}
\Gb_{Y^*_{1:n}}(x)=\left(\prod\limits_{i=1}^n p_i\right)~ C^*\left(\Fb(x,\la_1),\ldots,\Fb(x,\la_n)\right).
\end{equation*}
Thus, by Definition \ref{def1} and comparing the survival functions of  $Y_{1:n}$ and $Y^*_{1:n}$ the proof is completed. 
\end{proof}
The following theorem provides the usual stochastic order between the smallest
claim amounts in two heterogeneous portfolios of risks with the common associated copulas.
\begin{theorem}\label{t2}
Let $ X_{\la_1},\ldots,X_{\la_n} $ ($ X_{\la^*_1},\ldots,X_{\la^*_n} $) be non-negative random variables with
$ X_{\la_i} \thicksim \Fb(x;\la_i)$ ($ X_{\la^*_i} \thicksim \Fb(x;\la^*_i)$), $i = 1, \ldots,n $, and the associated copula $C_R$. Further, suppose that $I_{p_1},\ldots,I_{p_n}$ ($I_{p^*_1},\ldots,I_{p^*_n}$) is a set of independent Bernoulli random variables, independent of the $X_{\la_i}$'s ($X_{\la^*_i}$'s), with ${\rm E}[I_{p_i}]=p_i$ (${\rm E}[I_{p^*_i}]=p^*_i$), $i=1,\ldots,n$.  Assume that the following conditions hold:
\begin{itemize}
\item[{\rm (i)}] $\Fb(x;\la)$ is increasing and concave in $\la$ for any $x\in \mathbb{R}_+$;
\item[{\rm (ii)}] $C_R$ is Schur-concave.
\end{itemize}
Then, we have
\begin{eqnarray*}
\prod\limits_{i=1}^n p^*_i \leq \prod\limits_{i=1}^n p_i,~ (\la_1,\ldots,\la_n) \mathop \preceq \limits^{{\mathop{\rm w}} }  (\la^*_1,\ldots,\la^*_n)  \Longrightarrow Y^*_{1:n} \leq_{{\rm st}}Y_{1:n}.
\end{eqnarray*}
\end{theorem}
\begin{proof}
Define $\Psi(\lab)=- C_R\left(\Fb(x;\la_1),\ldots,\Fb(x;\la_n)\right)$. Based on the condition (ii) and the nature of copula, $-C_R$ is decreasing and Schur-convex. So, the condition (i) and Lemma \ref{mkl2} imply that $\Psi$ is decreasing and Schur-convex in $\lab$. Thus, using the Lemma \ref{mkl}, $\lab \mathop \preceq \limits^{{\mathop{\rm w}} } \lab^*$ implies $\Psi(\lab)\leq \Psi(\lab^*)$. Hence, the condition $\prod_{i=1}^n p^*_i \leq \prod_{i=1}^n p_i$ and the relation \eqref{dist1} complete the proof.
\end{proof}
The following theorem provides the ordering of the smallest
claim amounts from two heterogeneous portfolios of risks with the different parameter vectors and different associated copulas.
\begin{theorem}\label{t3}
Let $ X_{\la_1},\ldots,X_{\la_n} $ ($ X_{\la^*_1},\ldots,X_{\la^*_n} $) be non-negative random variables with
$ X_{\la_i} \thicksim \Fb(x;\la_i)$ ($ X_{\la^*_i} \thicksim \Fb(x;\la^*_i)$), $i = 1, \ldots,n $, and associated copula $C_R$ ($C^*_R$). Further, suppose that $I_{p_1},\ldots,
I_{p_n}$  ($I_{p^*_1},\ldots,I_{p^*_n} $) is a set of independent Bernoulli random variables, independent of the $X_{\la_i}$'s, with ${\rm E}[I_{p_i}]=p_i$ (${\rm E}[I_{p^*_i}]=p^*_i$), $i=1,\ldots,n$. Assume that the following conditions hold:
\begin{itemize}
\item[{\rm (i)}] $\Fb(x;\la)$ is increasing and concave in $\la$ for any $x\in \mathbb{R}_+$;
\item[{\rm (ii)}] $C^*_R$ is Schur-concave.
\end{itemize}
Then, we have
\begin{eqnarray*}
\prod\limits_{i=1}^n p^*_i \leq \prod\limits_{i=1}^n p_i,~ (\la_1,\ldots,\la_n) \mathop \preceq \limits^{{\mathop{\rm w}} }  (\la^*_1,\ldots,\la^*_n) ,~C^*_R\prec C_R \Longrightarrow Y^*_{1:n} \leq_{{\rm st}}Y_{1:n}.
\end{eqnarray*}
\end{theorem}
\begin{proof}
Suppose that $V^{C_R}_{\lab,\pb}$ denotes the smallest of the variables $Y_i=I_{p_i} X_{\la_i}$, $i=1,\ldots,n$, where $(X_{\la_1},\ldots,X_{\la_n})$ has the survival copula $C_R$. It is easily seen that $Y^*_{1:n}\mathop = \limits^{{\mathop{\rm st}} }V^{C^*_R}_{\lab^*,\pb^*}$ and $Y_{1:n}\mathop = \limits^{{\mathop{\rm st}} }V^{C_R}_{\lab,\pb}$. On the other hand, Theorem \ref{t1} and Theorem \ref{t2} imply that $V^{C^*_R}_{\lab,\pb} \leq_{{\rm st}}V^{C_R}_{\lab,\pb}$ and $V^{C^*_R}_{\lab^*,\pb^*} \leq_{{\rm st}}V^{C^*_R}_{\lab,\pb}$, respectively. Hence, the required result is obtained.
\end{proof}

Generally, Theorem \ref{t3} considers the comparison of the smallest claim amounts arising from two portfolios, in the sense of the usual stochastic order. But its result can be obtained in the sense of the stronger orders under some particular cases. The proportional hazard rate (PHR) model is an important model in reliability theory, actuarial science and other fields; see for example Cox\cite{cox}, Finkelstein\cite{fin}, Kumar and Klefsj\"{o}\cite{kukl} and Balakrishnan et al.\cite{baet}.  $X_{\la}$ is said to follow PHR model, if its survival function can be expressed as $\Fb(x;\la)=[\Fb(x)]^{\la}$, where $\Fb(x)$ is the baseline survival function and $\la>0$.

Recently, Li and Li\cite{lili2} compared $Y_{1:n}$ and $Y^*_{1:n}$ in the sense of the hazard rate order, whenever $ X_{\la_i} \thicksim \Fb(x;\la_i)=\Fb^{\la_i}(x)$ ($ X_{\la^*_i} \thicksim \Fb(x;\la^*_i)=\Fb^{\la^*_i}(x)$), for $i = 1, \ldots,n $, and they share a common Gumbel-Hougaard survival copula, which first introduced by Gumbel\cite{gum2}, of the form $$C_{R}(\ub)=\exp\left(-\left(\sum_{i=1}^n(-\log u_i)^{\theta}\right)^{1/\theta}\right),$$ for $\theta \in [1,\infty)$. They presented a characterization on the hazard rate order of $Y_{1:n}$ as the following lemma.

\begin{lemma}[Li and Li\cite{lili2}]\label{l9}
Let $ X_{\la_1},\ldots,X_{\la_n} $ ($ X_{\la^*_1},\ldots,X_{\la^*_n} $) be non-negative random variables with
$ X_{\la_i} \thicksim \Fb(x;\la_i)=\Fb^{\la_i}(x)$ ($ X_{\la^*_i} \thicksim \Fb(x;\la^*_i)=\Fb^{\la^*_i}(x)$), $i = 1, \ldots,n $, and associated Gumbel-Hougaard copula. Further, suppose that $I_{p_1},\ldots,
I_{p_n}$  ($I_{p^*_1},\ldots,I_{p^*_n} $) is a set of independent Bernoulli random variables, independent of the $X_{\la_i}$'s, with ${\rm E}[I_{p_i}]=p_i$ (${\rm E}[I_{p^*_i}]=p^*_i$), $i=1,\ldots,n$. Then, we have
\begin{eqnarray*}
\prod\limits_{i=1}^n p^*_i \leq \prod\limits_{i=1}^n p_i,~\sum_{i=1}^n \la^{\theta}_i \leq \sum_{i=1}^n {\la^*}^{\theta}_i \Longleftrightarrow Y^*_{1:n} \leq_{{\rm hr}}Y_{1:n}.
\end{eqnarray*}
\end{lemma}

The likelihood ratio order of $Y_{1:n}$ can also be characterized under some additional assumptions. The following theorems represent this fact.

\begin{theorem}\label{t15}
Under the setup of Lemma \ref{l9}, assume that $\prod_{i=1}^n p_i = \prod_{i=1}^n p^*_i$. Then, we have
\begin{eqnarray*}
\sum_{i=1}^n \la^{\theta}_i =\sum_{i=1}^n {\la^*}^{\theta}_i \Longleftrightarrow Y^*_{1:n} \leq_{{\rm lr}}Y_{1:n}.
\end{eqnarray*}
\end{theorem}

\begin{proof}
It can be easily verified that the ratio of density functions can be written as follows:
\begin{equation}\label{eq3}
\frac{g_{Y_{1:n}}(x)}{g_{Y^*_{1:n}}(x)}=\frac{1-\prod_{i=1}^n p_i}{1-\prod_{i=1}^n p^*_i} I_{[x=0]}+\frac{\prod_{i=1}^n p_i}{\prod_{i=1}^n p^*_i} \left(\frac{\sum_{i=1}^n \la^{\theta}_i}{\sum_{i=1}^n {\la^*}^{\theta}_i}\right)^{1/\theta} [\Fb(x)]^{(\sum_{i=1}^n \la^{\theta}_i)^{1/\theta}-(\sum_{i=1}^n {\la^*}^{\theta}_i)^{1/\theta}}I_{[x>0]},
\end{equation}
where, $I_A$ denotes the indicator function. Under the assumption $\prod_{i=1}^n p_i = \prod_{i=1}^n p^*_i$, we have that
\begin{equation*}
\frac{g_{Y_{1:n}}(x)}{g_{Y^*_{1:n}}(x)}=I_{[x=0]}+ \left(\frac{\sum_{i=1}^n \la^{\theta}_i}{\sum_{i=1}^n {\la^*}^{\theta}_i}\right)^{1/\theta} [\Fb(x)]^{(\sum_{i=1}^n \la^{\theta}_i)^{1/\theta}-(\sum_{i=1}^n {\la^*}^{\theta}_i)^{1/\theta}}I_{[x>0]}.
\end{equation*}
Clearly, $\frac{g_{Y_{1:n}}(x)}{g_{Y^*_{1:n}}(x)}$ is increasing in $x\geq 0$, if and only if $\left(\frac{\sum_{i=1}^n \la^{\theta}_i}{\sum_{i=1}^n {\la^*}^{\theta}_i}\right)^{1/\theta} [\Fb(0)]^{(\sum_{i=1}^n \la^{\theta}_i)^{1/\theta}-(\sum_{i=1}^n {\la^*}^{\theta}_i)^{1/\theta}} \geq \frac{g_{Y_{1:n}}(0)}{g_{Y^*_{1:n}}(0)}=1$ and $[\Fb(x)]^{(\sum_{i=1}^n \la^{\theta}_i)^{1/\theta}-(\sum_{i=1}^n {\la^*}^{\theta}_i)^{1/\theta}}$ is increasing in $x\geq 0$. The former is equivalent to $\sum_{i=1}^n \la^{\theta}_i \geq \sum_{i=1}^n {\la^*}^{\theta}_i$ and the latter is equivalent to $\sum_{i=1}^n \la^{\theta}_i \leq \sum_{i=1}^n {\la^*}^{\theta}_i$.
\end{proof}

\begin{theorem}\label{t16}
Under the setup of Lemma \ref{l9}, assume that $\sum_{i=1}^n \la^{\theta}_i =\sum_{i=1}^n {\la^*}^{\theta}_i$. Then, we have
\begin{eqnarray*}
\prod_{i=1}^n p^*_i \leq \prod_{i=1}^n p_i \Longleftrightarrow Y^*_{1:n} \leq_{{\rm lr}}Y_{1:n}.
\end{eqnarray*}
\end{theorem}

\begin{proof}
By the assumption $\sum_{i=1}^n \la^{\theta}_i = \sum_{i=1}^n {\la^*}^{\theta}_i$, the relation \eqref{eq3} can be rewritten as
\begin{equation*}
\frac{g_{Y_{1:n}}(x)}{g_{Y^*_{1:n}}(x)}=\frac{1-\prod_{i=1}^n p_i}{1-\prod_{i=1}^n p^*_i} I_{[x=0]}+\frac{\prod_{i=1}^n p_i}{\prod_{i=1}^n p^*_i} I_{[x>0]}.
\end{equation*}
Thus,  $\frac{g_{Y_{1:n}}(x)}{g_{Y^*_{1:n}}(x)}$ is increasing in $x\geq 0$, if and only if $\frac{\prod_{i=1}^n p_i}{\prod_{i=1}^n p^*_i} \geq \frac{1-\prod_{i=1}^n p_i}{1-\prod_{i=1}^n p^*_i}$, which is equivalent to $\prod_{i=1}^n p_i \geq \prod_{i=1}^n p^*_i$.
\end{proof}

%Note that the condition (iii) of Theorem \ref{t3} is not a complex condition in application. Archimedean copula is one of the most usable structures in the literature. The following theorem provides a comparison between the largest claim amounts in two heterogeneous portfolio of risks, in the presence of Archimedean copula.

%Note that the condition (iii) of Theorem \ref{t3} is not a complex condition in application. The following lemma provides a condition that satisfies the condition (iii) of Theorem \ref{t3}.

%\begin{lemma}\label{l4}
%The condition
%\begin{itemize}
%\item[{\rm (iii)}$^{\prime}$] $C$ is an Archimedean copula which is PQD,
%\end{itemize}
%satisfies the condition
%\begin{itemize}
%\item[{\rm (iii)}] $C$ is PQD and $\frac{\partial C(u,v)}{\partial u}\geq \frac{\partial C(u,v)}{\partial v}$, for all $0\leq u\leq v\leq 1$.
%\end{itemize}
%\end{lemma}
%\begin{proof}
%In the view of Lemma \ref{l6} and Lemma \ref{l7}, the proof is completed.
%\end{proof}

\subsection{Bounds for the survival function of the smallest claim amount}\label{sub2}
Obtaining some bounds for $\Gb_{Y_{1:n}}(x)$ can be included the important informations for the insurance companies. The following theorem presents useful lower and upper bounds for $\Gb_{Y_{1:n}}(x)$ when the insurance company knows the associated copula, exactly. 
\begin{theorem}\label{t4}
Let $ X_{\la_1},\ldots,X_{\la_n} $ be non-negative random variables with
$ X_{\la_i} \thicksim \Fb(x;\la_i)$, $i = 1, \ldots,n $, and the associated copula $C_R$. Further, suppose that $I_{p_1},\ldots,
I_{p_n}$ is a set of independent Bernoulli random variables, independent of the $X_{\la_i}$'s, with ${\rm E}[I_{p_i}]=p_i$, $i=1,\ldots,n$. Assume that the following conditions hold:
\begin{itemize}
\item[{\rm (i)}] $\Fb(x;\la)$ is increasing and concave in $\la$ for any $x\in \mathbb{R}_+$;
\item[{\rm (ii)}] $C_R$ is Schur-concave.
\end{itemize}
Then, we have
\begin{eqnarray*}
\left(\prod\limits_{i=1}^n p_i\right) \delta_{C_R}\left(\Fb(x;\la_{1:n})\right)\leq \Gb_{Y_{1:n}}(x)\leq \left(\prod\limits_{i=1}^n p_i\right) \delta_{C_R}\left(\Fb(x;\bar{\la})\right).
\end{eqnarray*}
\end{theorem}

\begin{proof}
The fact that $(\bar{\la},\ldots,\bar{\la})\mathop \preceq \limits^{{\mathop{\rm w}} }(\la_1,\ldots,\la_n) \mathop \preceq \limits^{{\mathop{\rm w}} }  (\la_{1:n},\ldots,\la_{1:n})$ and Theorem \ref{t3} imply the required result.
\end{proof}

Usually, the associated copula is unknown for a company, while the sign of dependency is well-known. Naturally, one may wonder whether presenting the lower and upper bounds for $\Gb_{Y_{1:n}}(x)$ is possible? The following theorem has a positive answer for this question.

\begin{theorem}\label{t5}
Under the setup of Theorem \ref{t4}, suppose that the following conditions hold:
\begin{itemize}
\item[{\rm (i)}] $\Fb(x;\la)$ is increasing and concave in $\la$ for any $x\in \mathbb{R}_+$;
\item[{\rm (ii)}] $C_R$ is {\rm PUOD}.
\end{itemize}
Then, we have
\begin{eqnarray}\label{eq2}
\left(\prod\limits_{i=1}^n p_i\right) \Fb^n(x;\la_{1:n})\leq \Gb_{Y_{1:n}}(x)\leq \left(\prod\limits_{i=1}^n p_i\right) \Fb(x;\la_{1:n}).
\end{eqnarray}
\end{theorem}

\begin{proof}
Let $C^*_R(\ub)=\prod_{i=1}^n u_i$ be the independent copula. Since $C^*_R$ is an Archimedean copula, so  Lemma \ref{l7} guarantees the Schur-concavity of $C^*_R$. Thus, using the fact that $(\la_1,\ldots,\la_n) \mathop \preceq \limits^{{\mathop{\rm w}} }  (\la_{1:n},\ldots,\la_{1:n})$, Theorem \ref{t3} and Definition \ref{def2}, we have
\begin{eqnarray*}
\Gb_{Y_{1:n}}(x)&=&\left(\prod\limits_{i=1}^n p_i\right)~ C_R\left(\Fb(x,\la_1),\ldots,\Fb(x,\la_n)\right)\\
&\geq& \left(\prod\limits_{i=1}^n p_i\right)~ C^*_R\left(\Fb(x,\la_{1:n}),\ldots,\Fb(x,\la_{1:n})\right)\\
&=&\left(\prod\limits_{i=1}^n p_i\right) \Fb^n(x;\la_{1:n}),
\end{eqnarray*}
which proves the first inequality in \eqref{eq2}. On the other hand, the Fr\'echet-Hoeffding upper bound in Lemma \ref{l6} and increasing property of $\Fb(x;\la)$ in $\la$,  immediately proves the second inequality in \eqref{eq2}. Hence, the proof is completed. 
\end{proof}
 
The proportional reversed hazard rate (PRHR) model, which introduced by Gupta et al.\cite{gup}, is a flexible family of distributions existing in reliability theory, and can be used in actuarial science and other fields.  $X_{\la}$ is said to follow PRHR model, if its distribution function can be expressed as $F(x;\la)=F^{\la}(x)$, where $F(x)$ is the baseline distribution function and $\la>0$.

The following corollaries provide a lower and upper bounds for the survival function of the smallest
claim amount in a portfolio of risks, whenever the marginal distributions of severities belonging to the PRHR model.

%\begin{theorem}\label{t6}
%Let $F(x;\la_i)=F^{\la_i}(x)$ and $F(x;\la^*_i)=F^{\la^*_i}(x)$, for $i=1,\ldots,n$. Under the setup of Theorem \ref{t3}, suppose that $C^*_R$ is Schur-concave. Then, we have
%\begin{eqnarray*}
%\prod\limits_{i=1}^n p^*_i \leq \prod\limits_{i=1}^n p_i,~ (\la_1,\ldots,\la_n) \mathop \preceq \limits^{{\mathop{\rm w}} }  (\la^*_1,\ldots,\la^*_n) ,~C^*\prec C \Longrightarrow Y^*_{1:n} \leq_{{\rm st}}Y_{1:n}.
%\end{eqnarray*}
%\end{theorem}
%\begin{proof}
%It is easily seen that $\Fb(x;\la)=1-F^{\la}(x)$ is increasing and concave in $\la$. Thus, Theorem \ref{t3} immediately completes the proof.
%\end{proof}

\begin{corollary}\label{t7}
Let $F(x;\la_i)=F^{\la_i}(x)$, for $i=1,\ldots,n$. Under the setup of Theorem \ref{t4}, suppose that $C_R$ is Schur-concave. Then, we have
\begin{eqnarray*}
\left(\prod\limits_{i=1}^n p_i\right) \delta_{C_R}\left(1-F^{\la_{1:n}}(x)\right)\leq \Gb_{Y_{1:n}}(x)\leq \left(\prod\limits_{i=1}^n p_i\right) \delta_{C_R}\left(1-F^{\bar{\la}}(x)\right).
\end{eqnarray*}
\end{corollary}
\begin{proof}
Clearly, $\Fb(x;\la)=1-F^{\la}(x)$ and $C_R$ satisfy the conditions (i) and (ii) of Theorem \ref{t4}, respectively. Hence, Theorem \ref{t4} completes the proof.
\end{proof}

\begin{corollary}\label{t8}
Let $F(x;\la_i)=F^{\la_i}(x)$, for $i=1,\ldots,n$. Under the setup of Theorem \ref{t4}, suppose that $C_R$ is {\rm PUOD}. Then, we have
\begin{eqnarray*}
\left(\prod\limits_{i=1}^n p_i\right) \left(1-F^{\la_{1:n}}(x)\right)^n \leq \Gb_{Y_{1:n}}(x)\leq \left(\prod\limits_{i=1}^n p_i\right) \left(1-F^{\la_{1:n}}(x)\right).
\end{eqnarray*}
\end{corollary}

\begin{proof}
It is clear that the conditions (i) and (ii) of Theorem \ref{t5} are satisfied. Hence, Theorem \ref{t5} implies the required result.
\end{proof}

The following example provides a numerical example to illustrate the validity of corollaries \ref{t7} and \ref{t8}.
\begin{example}\label{ex1}
Let $X_{\la_i}\thicksim F(x;\la_i)=(1-e^{-x})^{\la_i} $, for $i = 1, 2,3$, with the associated Frank copula, which introduced by Frank\cite{frank}, of the form 
$$C_{R}(u_1,u_2,u_3)=-\frac{1}{\theta}\log \left(1+\frac{(e^{-\theta u_1}-1)(e^{-\theta u_2}-1)(e^{-\theta u_3}-1)}{(e^{-\theta}-1)^2}\right),$$ where $\theta\in(0,\infty)$. Further, suppose that $I_{p_1}, I_{p_2}, I_{p_3}$ is a set of independent Bernoulli random variables, independent of the $X_{\la_i}$'s, with ${\rm E}[I_{p_i}]=p_i$ , for $i=1,2,3$. We take $(\la_1,\la_2,\la_3)=(3,6,2)$, $(p_1,p_2,p_3)=(0.5,0.6,0.1)$ and $\theta=5$. According to Nelsen\cite{nel}, $C_{R}$ is an Archimedean copula and according to Lemma \ref{l7} is Schur-concave. Also it is a {\rm PUOD} copula. Thus, the conditions of theorems \ref{t7} and \ref{t8} are satisfied. Figure \ref{fig1} represents the plots of the survival function of the smallest claim amounts and the proposed bounds in corollaries \ref{t7} and \ref{t8}.
\end{example}

\begin{figure}[ht]
\centerline{\includegraphics[width=7cm,height=7cm]{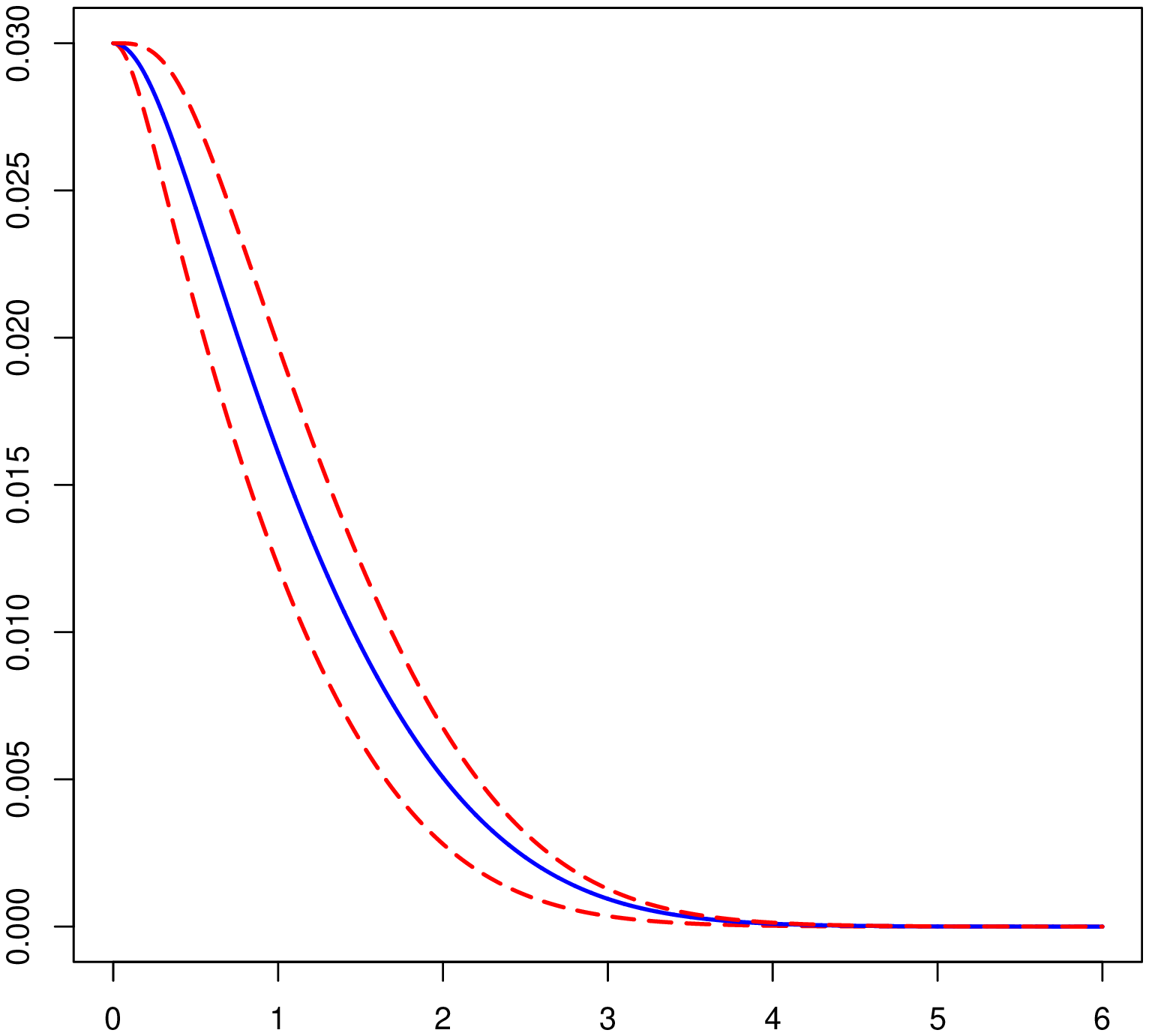}
\includegraphics[width=7cm,height=7cm]{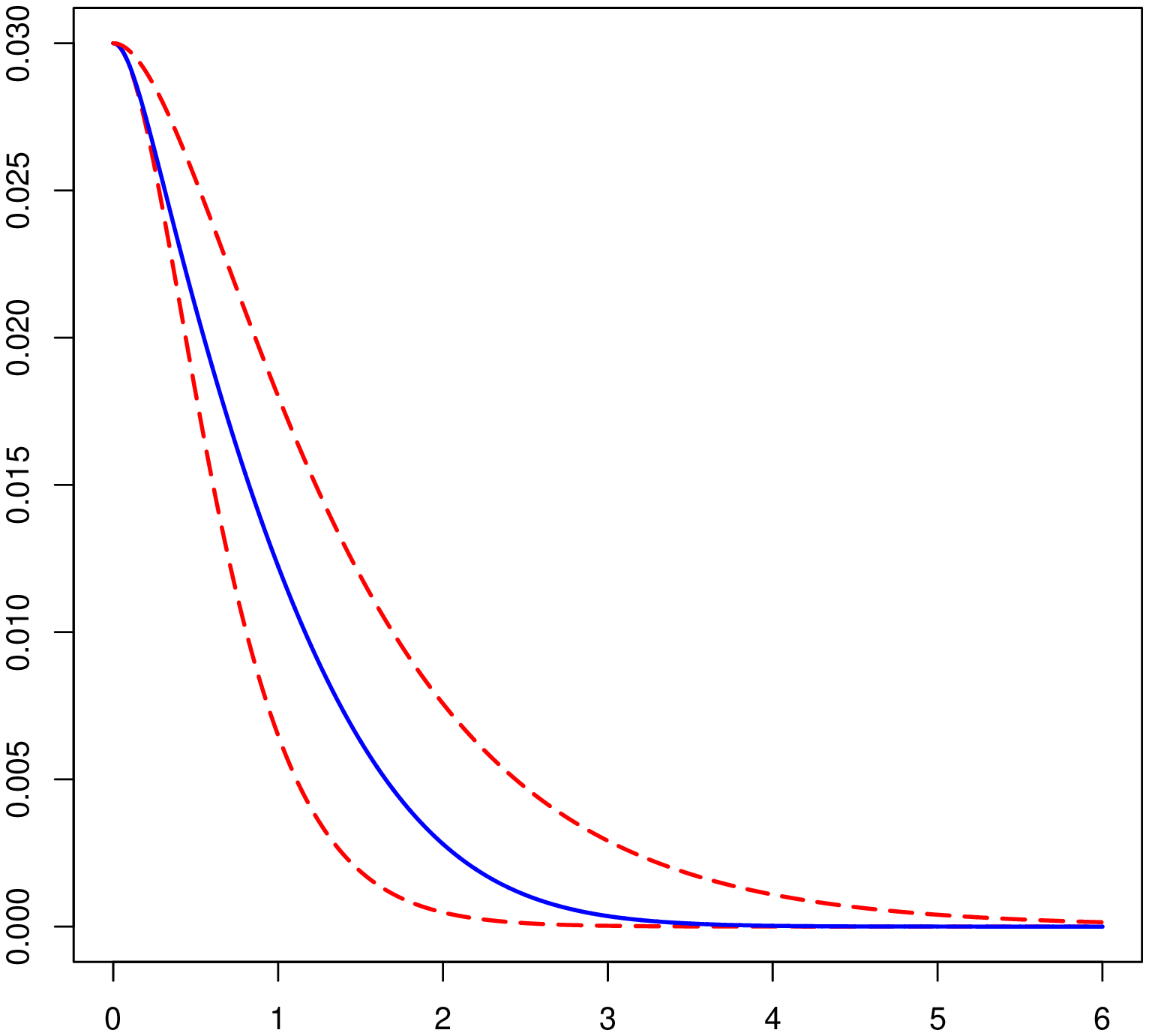}}
\vspace{-0.8cm} \caption{\small{ Plots of the survival function of the smallest claim amounts and the proposed bounds in Corollary \ref{t7} (left) and Corollary \ref{t8} (right) for Example \ref{ex1}.}}\label{fig1}
\end{figure}

Recently, Barmalzan et al.\cite{bar2018} considerd the Marshall-Olkin extended exponential distribution as the claim amount distribution in portfolios of risks and compared aggregate claim amounts in two heterogeneous portfolios. Marshall-Olkin distribution, which introduced by Marshall and Olkin\cite{marolk} is a spacial case of a wide range family of distributions, called Harris family. Aly and Benkherouf\cite{aly} used the Harris distribution introduced by Harris\cite{harris}, and generated the Harris family. $X_{\la}$ is said to follow the Harris family, if its survival function is given by
\begin{equation*}
\Fb(x;\la,\theta)=\left(\frac{\la \Fb^{\theta}(x)}{1-(1-\la)\Fb^{\theta}(x)}\right)^{1/\theta},\quad  \la>0,~\theta>0,
 \end{equation*}
where, $\Fb(x)$ is the baseline survival function.

%\begin{theorem}\label{t9}
%Let $\Fb(x;\la_i,\theta)=\left(\frac{\la_i \Fb^{\theta}(x)}{1-(1-\la_i)\Fb^{\theta}(x)}\right)^{1/\theta}$ and $\Fb(x;\la^*_i,\theta)=\left(\frac{\la^*_i \Fb^{\theta}(x)}{1-(1-\la^*_i)\Fb^{\theta}(x)}\right)^{1/\theta}$, for $\theta\geq 1$ and $i=1,\ldots,n$. Under the setup of Theorem \ref{t3}, suppose that $C^*$ is Schur-concave. Then, we have
%\begin{eqnarray*}
%\prod\limits_{i=1}^n p^*_i \leq \prod\limits_{i=1}^n p_i,~ (\la_1,\ldots,\la_n) \mathop \preceq \limits^{{\mathop{\rm w}} }  (\la^*_1,\ldots,\la^*_n) ,~C^*\prec C \Longrightarrow Y^*_{1:n} \leq_{{\rm st}}Y_{1:n}.
%\end{eqnarray*}
%\end{theorem}
%\begin{proof}
%By simplification, the first and second partial derivatives of $\Fb(x;\la,\theta)=\left(\frac{\la \Fb^{\theta}(x)}{1-(1-\la)\Fb^{\theta}(x)}\right)^{1/\theta}$ are obtained as follows:
%\begin{eqnarray*}
%\frac{\partial \Fb(x;\la,\theta)}{\partial \la}&=&\frac{\Fb(x;\la,\theta)}{\la \left(1-(1-\la) \Fb^{\theta}(x)\right)}\geq 0,\\
%\frac{\partial^2 \Fb(x;\la,\theta)}{\partial \la^2}&=&\frac{1-\Fb^{\theta}(x)}{\theta}\frac{\Fb(x;\la,\theta)}{\la^2 \left(1-(1-\la) \Fb^{\theta}(x)\right)^2}\left(\left(1-\Fb^{\theta}(x)\right)\left(\frac{1}{\theta}-1\right)-2\la \Fb^{\theta}(x) \right)\leq 0,
%\end{eqnarray*}
%which, the first inequality is clear and the second is due to the assumption $\theta\geq 1$. Thus, $\Fb(x;\la,\theta)$ is increasing and concave in $\la$. Hence, in the view of Theorem \ref{t3} the desired result is obtained.
%\end{proof}

The following corollaries provide a lower and upper bounds for the survival function of the smallest
claim amount in a portfolio of risks, whenever the marginal distributions of severities belonging to the Harris family.

\begin{corollary}\label{t10}
Let $\Fb(x;\la_i,\theta)=\left(\frac{\la_i \Fb^{\theta}(x)}{1-(1-\la_i)\Fb^{\theta}(x)}\right)^{1/\theta}$, for $\theta\geq 1$ and $i=1,\ldots,n$. Under the setup of Theorem \ref{t4}, suppose that $C_R$ is Schur-concave. Then, we have
\begin{eqnarray*}
\left(\prod\limits_{i=1}^n p_i\right) \delta_{C_R}\left(\left(\frac{\la_{1:n} \Fb^{\theta}(x)}{1-(1-\la_{1:n})\Fb^{\theta}(x)}\right)^{1/\theta}\right)\leq \Gb_{Y_{1:n}}(x)\leq \left(\prod\limits_{i=1}^n p_i\right) \delta_{C_R}\left(\left(\frac{\bar{\la} \Fb^{\theta}(x)}{1-(1-\bar{\la})\Fb^{\theta}(x)}\right)^{1/\theta}\right).
\end{eqnarray*}
\end{corollary}
\begin{proof}
By simplification, the first and second partial derivatives of $\Fb(x;\la,\theta)=\left(\frac{\la \Fb^{\theta}(x)}{1-(1-\la)\Fb^{\theta}(x)}\right)^{1/\theta}$ are obtained as follows:
\begin{eqnarray*}
\frac{\partial \Fb(x;\la,\theta)}{\partial \la}&=&\frac{\Fb(x;\la,\theta)}{\la \left(1-(1-\la) \Fb^{\theta}(x)\right)}\geq 0,\\
\frac{\partial^2 \Fb(x;\la,\theta)}{\partial \la^2}&=&\frac{1-\Fb^{\theta}(x)}{\theta}\frac{\Fb(x;\la,\theta)}{\la^2 \left(1-(1-\la) \Fb^{\theta}(x)\right)^2}\left(\left(1-\Fb^{\theta}(x)\right)\left(\frac{1}{\theta}-1\right)-2\la \Fb^{\theta}(x) \right)\leq 0,
\end{eqnarray*}
which, the first inequality is clear and the second is due to the assumption $\theta\geq 1$. Thus, $\Fb(x;\la,\theta)$ is increasing and concave in $\la$. Hence, in the view of Theorem \ref{t4} the desired result is obtained.
\end{proof}

\begin{corollary}\label{t11}
Let $\Fb(x;\la_i,\theta)=\left(\frac{\la_i \Fb^{\theta}(x)}{1-(1-\la_i)\Fb^{\theta}(x)}\right)^{1/\theta}$, for $\theta\geq 1$ and $i=1,\ldots,n$. Under the setup of Theorem \ref{t4}, suppose that $C_R$ is {\rm PUOD}. Then, we have
\begin{eqnarray*}
\left(\prod\limits_{i=1}^n p_i\right) \left(\frac{\la_{1:n} \Fb^{\theta}(x)}{1-(1-\la_{1:n})\Fb^{\theta}(x)}\right)^{n/\theta} \leq \Gb_{Y_{1:n}}(x)\leq \left(\prod\limits_{i=1}^n p_i\right) \left(\frac{\la_{1:n} \Fb^{\theta}(x)}{1-(1-\la_{1:n})\Fb^{\theta}(x)}\right)^{1/\theta}.
\end{eqnarray*}
\end{corollary}

\begin{proof}
Obviously, the conditions (i) and (ii) of Theorem \ref{t5} are satisfied. Hence, Theorem \ref{t5} completes the proof.
\end{proof}

The following example provides a numerical example to illustrate the validity of corollaries \ref{t10} and \ref{t11}.
\begin{example}\label{ex2}
Let $X_{\la_i}\thicksim \Fb(x;\la_i)=\left(\frac{\la_i e^{-3x^2}}{1-(1-\la_i)e^{-3x^2}}\right)^{1/3} $, for $i = 1, 2,3$, with the associated Clayton copula, which introduced by Clayton\cite{clay}, of the form $$C_{R}(u_1,u_2,u_3)=(u^{-\theta}_1+u^{-\theta}_2+u^{-\theta}_3-2)^{-1/\theta},$$
 where $\theta\in(0,\infty)$. Further, suppose that $I_{p_1}, I_{p_2}, I_{p_3}$ is a set of independent Bernoulli random variables, independent of the $X_{\la_i}$'s, with ${\rm E}[I_{p_i}]=p_i$ , for $i=1,2,3$. We take $(\la_1,\la_2,\la_3)=(3,5,1)$, $(p_1,p_2,p_3)=(0.2,0.3,0.2)$ and $\theta=3$. According to Nelsen\cite{nel}, $C_{R}$ is an Archimedean copula and according to Lemma \ref{l7} is Schur-concave. Also it is a {\rm PUOD} copula. Thus, the conditions of theorems \ref{t10} and \ref{t11} are satisfied. Figure \ref{fig2} represents the plots of the survival function of the smallest claim amounts and the proposed bounds in corollaries \ref{t10} and \ref{t11}.
\end{example}

\begin{figure}[ht]
\centerline{\includegraphics[width=7cm,height=7cm]{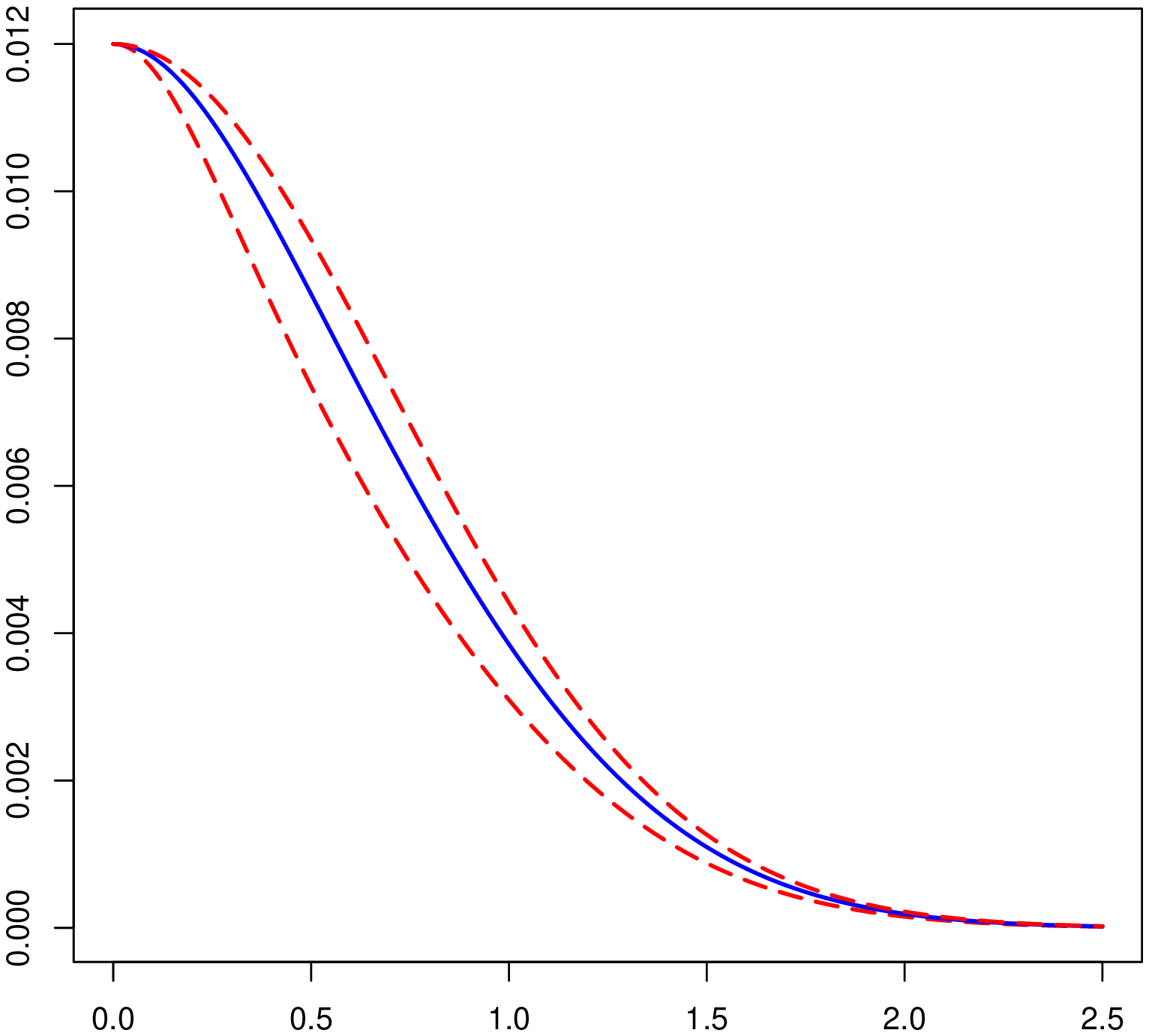}
\includegraphics[width=7cm,height=7cm]{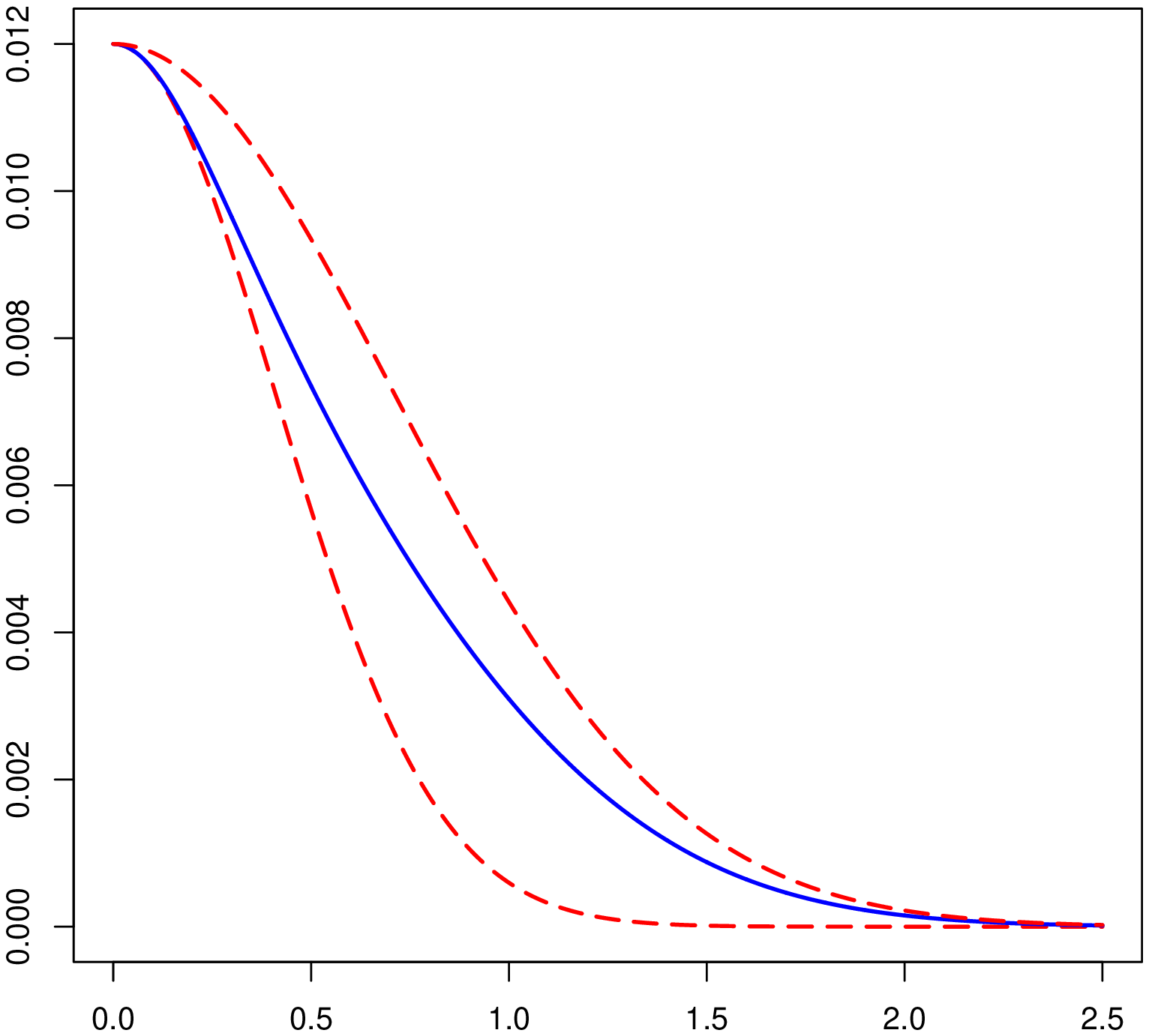}}
\vspace{-0.8cm} \caption{\small{ Plots of the survival function of the smallest claim amounts and the proposed bounds in Corollary \ref{t10} (left) and Corollary \ref{t11} (right) for Example \ref{ex2}.}}\label{fig2}
\end{figure}

Hami Golzar et al.\cite{hami} proposed the Lomax-exponential distribution, which is a proper model for right-skewed, approximately symmetric or reversed-J
shape populations. Due to simplicity and flexibility, it is a good alternative for positive populations and especially claim amounts in portfolios of risks. Recently, Nadeb and Torabi\cite{nato} discussed some stochastic comparisons of series systems with independent heterogeneous Lomax-exponential components. 
$X$ has the Lomax-exponential distribution with the positive parameters $\al$, $\be$ and $\la$, denoted by $X \thicksim {\rm LE}(\al,\beta,\la)$, if its survival function is given by
\begin{equation*}
\Fb(x;\al,\be,\la)=\left(\frac{\la}{e^{\be x}+\la-1}\right)^{\al},\quad x\in \mathbb{R}_+.
\end{equation*}

%\begin{theorem}\label{t12}
%Let $X_{\la_i}\thicksim {\rm LE}(\al,\beta,\la_i)$ and $X_{\la^*_i}\thicksim {\rm LE}(\al,\beta,\la^*_i)$, for $\al \leq 1$, $\be >0$ and $i=1,\ldots,n$. Under the setup of Theorem \ref{t3}, suppose that $C^*$ is Schur-concave. Then, we have
%\begin{eqnarray*}
%\prod\limits_{i=1}^n p^*_i \leq \prod\limits_{i=1}^n p_i,~ (\la_1,\ldots,\la_n) \mathop \preceq \limits^{{\mathop{\rm w}} }  (\la^*_1,\ldots,\la^*_n) ,~C^*\prec C \Longrightarrow Y^*_{1:n} \leq_{{\rm st}}Y_{1:n}.
%\end{eqnarray*}
%\end{theorem}
%\begin{proof}
%The first and second partial derivatives of $\Fb(x;\al,\beta,\la)=\left(\frac{\la}{e^{\be x}+\la-1}\right)^{\al}$ are given by
%\begin{eqnarray*}
%\frac{\partial \Fb(x;\al,\beta,\la)}{\partial \la}&=&\frac{\al (e^{\be x}-1)}{\la (e^{\be x}+\la-1)}\Fb(x;\al,\beta,\la)\geq 0,\\
%\frac{\partial^2 \Fb(x;\al,\beta,\la)}{\partial \la^2}&=&\frac{\al (e^{\be x}-1)}{\la^2 (e^{\be x}+\la-1)^2}\Fb(x;\al,\beta,\la)\left((\al-1)(e^{\be x}-1)-2\la\right)\leq 0,
%\end{eqnarray*}
%which, the first inequality is clear and the second is due to the assumption $\al\leq 1$. Thus, $\Fb(x;\la,\theta)$ is increasing and concave in $\la$. Hence, in the view of Theorem \ref{t3} the desired result is obtained.
%\end{proof}

The following corollaries provide some bounds for the survival function of the smallest
claim amount in a portfolio of risks, when the marginal distributions of severities are Lomax-exponential.

\begin{corollary}\label{t13}
Let $X_{\la_i}\thicksim {\rm LE}(\al,\beta,\la_i)$, for $\al\leq 1$, $\be>0$ and $i=1,\ldots,n$. Under the setup of Theorem \ref{t4}, suppose that $C_R$ is Schur-concave. Then, we have
\begin{eqnarray*}
\left(\prod\limits_{i=1}^n p_i\right) \delta_{C_R}\left(\left(\frac{\la_{1:n}}{e^{\be x}+\la_{1:n}-1}\right)^{\al}\right)\leq \Gb_{Y_{1:n}}(x)\leq \left(\prod\limits_{i=1}^n p_i\right) \delta_{C_R}\left(\left(\frac{\bar{\la}}{e^{\be x}+\bar{\la}-1}\right)^{\al}\right).
\end{eqnarray*}
\end{corollary}
\begin{proof}
The first and second partial derivatives of $\Fb(x;\al,\beta,\la)=\left(\frac{\la}{e^{\be x}+\la-1}\right)^{\al}$ are given by
\begin{eqnarray*}
\frac{\partial \Fb(x;\al,\beta,\la)}{\partial \la}&=&\frac{\al (e^{\be x}-1)}{\la (e^{\be x}+\la-1)}\Fb(x;\al,\beta,\la)\geq 0,\\
\frac{\partial^2 \Fb(x;\al,\beta,\la)}{\partial \la^2}&=&\frac{\al (e^{\be x}-1)}{\la^2 (e^{\be x}+\la-1)^2}\Fb(x;\al,\beta,\la)\left((\al-1)(e^{\be x}-1)-2\la\right)\leq 0,
\end{eqnarray*}
which, the first inequality is clear and the second is due to the assumption $\al\leq 1$. Thus, $\Fb(x;\la,\theta)$ is increasing and concave in $\la$. Hence, in the view of Theorem \ref{t3} the desired result is obtained.
\end{proof}

\begin{corollary}\label{t14}
Let $X_{\la_i}\thicksim {\rm LE}(\al,\beta,\la_i)$, for $\al\leq 1$, $\be>0$ and $i=1,\ldots,n$. Under the setup of Theorem \ref{t4}, suppose that $C_R$ is {\rm PUOD}. Then, we have
\begin{eqnarray*}
\left(\prod\limits_{i=1}^n p_i\right) \left(\frac{\la_{1:n}}{e^{\be x}+\la_{1:n}-1}\right)^{n \al} \leq \Gb_{Y_{1:n}}(x)\leq \left(\frac{\la_{1:n}}{e^{\be x}+\la_{1:n}-1}\right)^{\al}.
\end{eqnarray*}
\end{corollary}

\begin{proof}
Applying Theorem \ref{t5}, the desired result is immediately obtained.
\end{proof}

The following example provides a numerical example to illustrate the validity of corollaries \ref{t13} and \ref{t14}.
\begin{example}\label{ex3}
Let $X_{\la_i}\thicksim {\rm LE}(0.1,3,\la_i) $, for $i = 1, 2,3$, with the associated Gumbel-Hougaard copula, of the form $$C_{R}(u_1,u_2,u_3)=\exp\bigg(-\left[(-\log u_1)^{\theta}+(-\log u_2)^{\theta}+(-\log u_3)^{\theta}\right]^{1/\theta}\bigg),$$ where $\theta\in [1,\infty)$. Further, suppose that $I_{p_1}, I_{p_2}, I_{p_3}$ is a set of independent Bernoulli random variables, independent of the $X_{\la_i}$'s, with ${\rm E}[I_{p_i}]=p_i$ , for $i=1,2,3$. We take $(\la_1,\la_2,\la_3)=(0.7,5,0.4)$, $(p_1,p_2,p_3)=(0.1,0.2,0.8)$ and $\theta=2$. According to Nelsen\cite{nel}, $C_{R}$ is an Archimedean copula and according to Lemma \ref{l7} is Schur-concave. Also it is a {\rm PUOD} copula. Thus, the conditions of corollaries \ref{t13} and \ref{t14} are satisfied. Figure \ref{fig3} represents the plots of the survival function of the smallest claim amounts and the proposed bounds in corollaries \ref{t13} and \ref{t14}.
\end{example}

\begin{figure}[ht]
\centerline{\includegraphics[width=7cm,height=7cm]{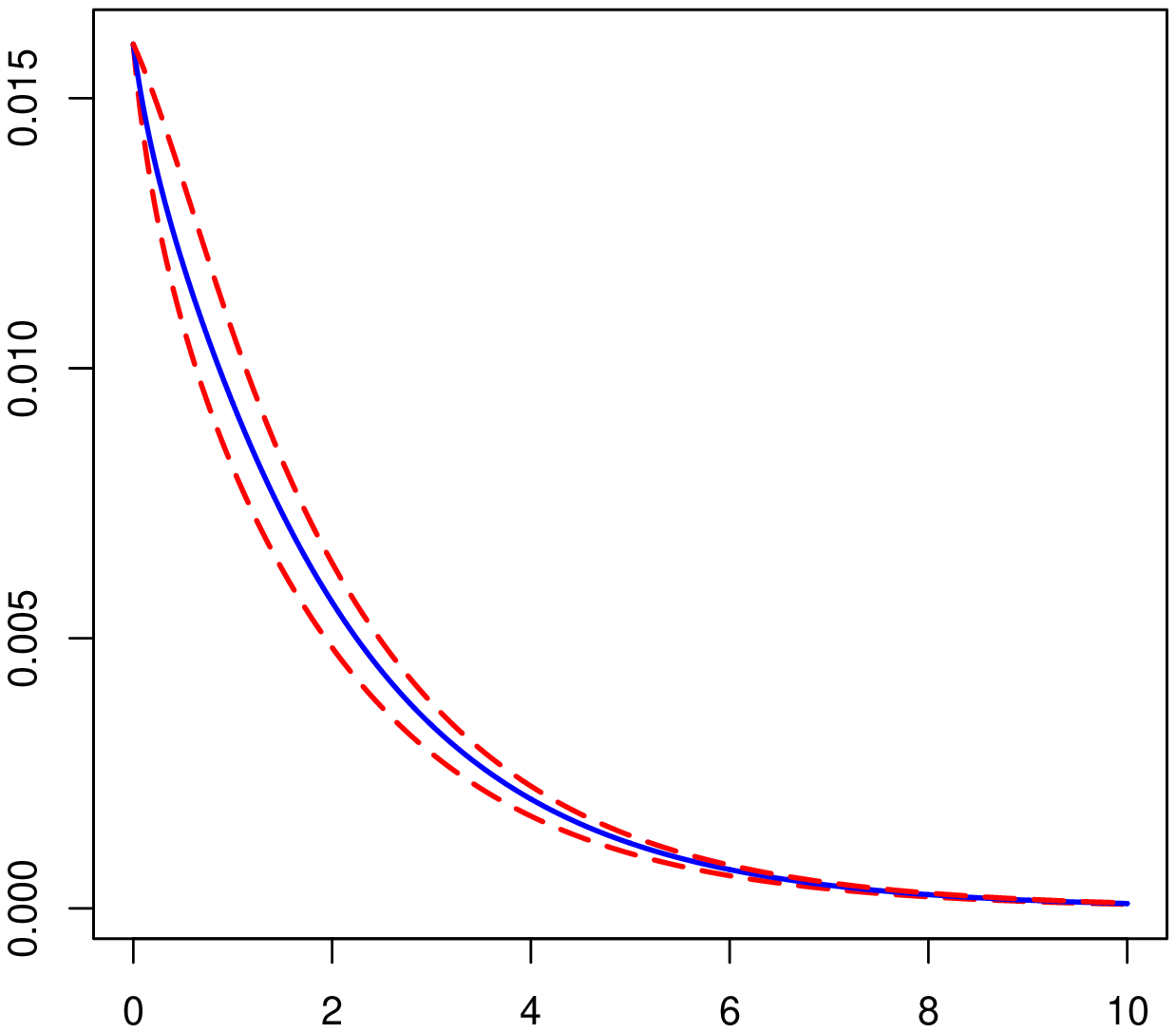}
\includegraphics[width=7cm,height=7cm]{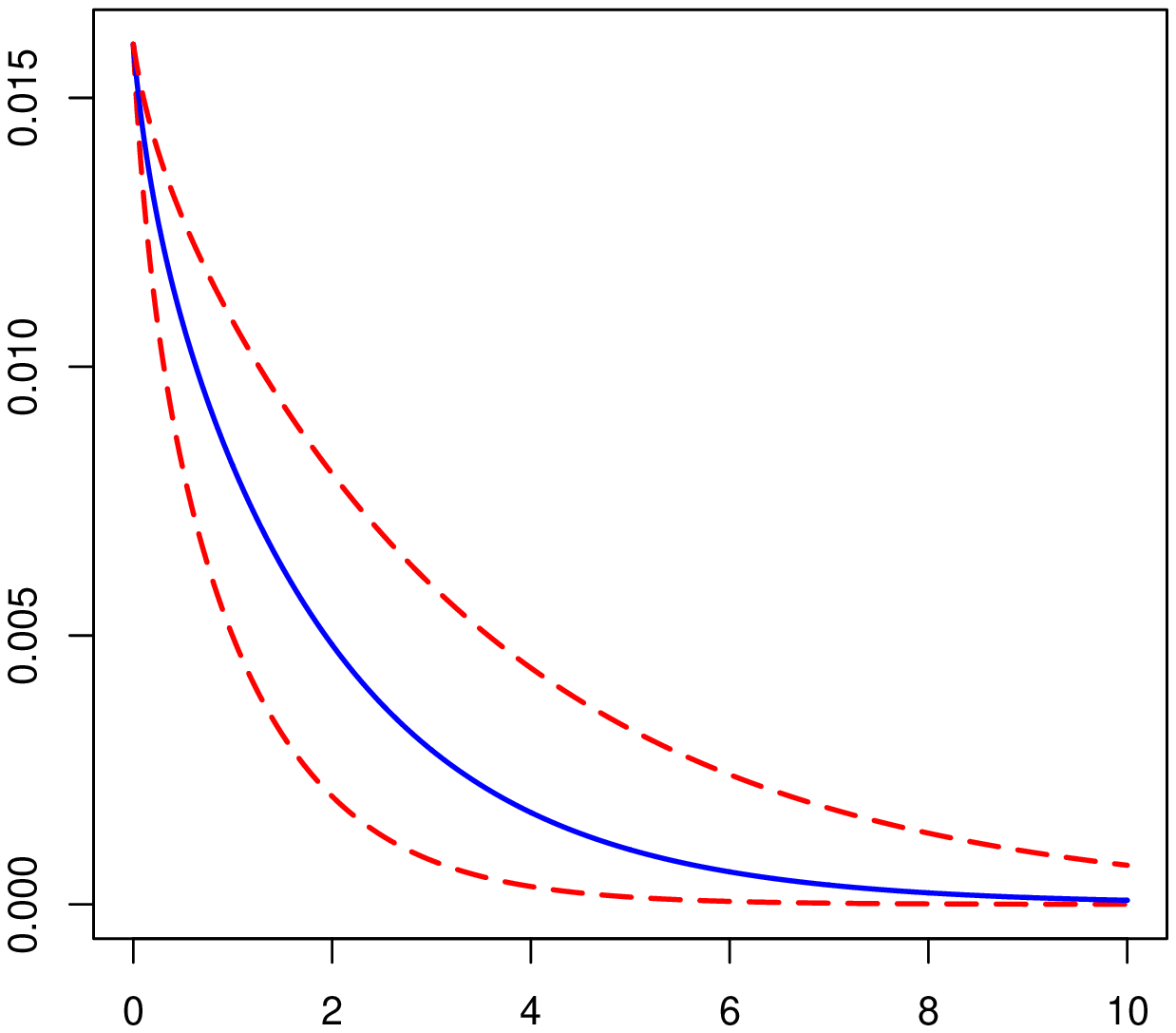}}
\vspace{-0.8cm} \caption{\small{ Plots of the survival function of the smallest claim amounts and the proposed bounds in Corollary \ref{t13} (left) and Corollary \ref{t14} (right) for Example \ref{ex3}.}}\label{fig3}
\end{figure}

\section*{Conclusion}
In this paper, under some certain conditions, we first discussed stochastic comparisons between the
smallest claim amounts under the assumption dependency of severities in the sense of usual and likelihood ratio orders in some general models. Next we present some helpful bounds for the survival function of the smallest claim amount in an interdependent heterogeneous portfolio. Also, some examples are served to illustrate the established results.

%%% generate the .bbl file with bibtex
%%% replace "yourbibfile" with the name of your .bib file
%%% (without .bib extension)
%\bibliographystyle{jstpip}
%\bibliography{yourbibfile}

\begin{thebibliography}{9}


%\bibitem[(1978)]{ali} Ali, M.M., Mikhail, N.N., and Haq, M.S. (1978). A class of bivariate distributions including the bivariate logistic. {\it Journal of Multivariate Analysis},  {\bf 8(3)}, 405-412.

\bibitem[(2011)]{aly} Aly, E.E.A.A, and Benkherouf, L. (2011). A new family of distributions based on probability generating functions. {\it Sankhya B}, {\bf 73(1)}, 70-82.

%\bibitem[(2011)]{arts} Aryal, G.R., and Tsokos, C.P. (2011). Transmuted Weibull distribution: A generalization of the Weibull probability distribution. {\it European Journal of Pure and Applied Mathematics},  {\bf 4(2)}, 89-102.


%\bibitem[(2015)]{bee} Balakrishnan, N., Haidari, A. and Masoumifard, K. (2015). Stochastic comparisons of series and parallel systems with generalized exponential components. {\it IEEE Transactions on Reliability},  {\bf 64(1)}, 333-348.

\bibitem[(2018)]{baet} Balakrishnan, N., Zhang, Y., and Zhao, P. (2018). Ordering the largest claim amounts and ranges from two sets of heterogeneous portfolios. {\it Scandinavian Actuarial Journal},  {\bf 2018(1)}, 23-41.

\bibitem[(2015)]{bana} Barmalzan, G., and Najafabadi, A.T.P. (2015). On the convex transform and right-spread orders of smallest claim amounts.  {\it Insurance: Mathematics and Economics},  {\bf 64}, 380-384.

\bibitem[(2015)]{bar1} Barmalzan, G., Najafabadi, A.T.P. and Balakrishnan, N. (2015). Stochastic comparison of aggregate claim amounts between two heterogeneous portfolios  and its applications. {\it Insurance: Mathematics and Economics},  {\bf 61}, 235-241.

\bibitem[(2016)]{bar3} Barmalzan, G., Najafabadi, A.T.P., and Balakrishnan, N. (2016). Likelihood ratio and dispersive orders for smallest order statistics and smallest claim amounts from heterogeneous Weibull sample. {\it Statistics and Probability Letters},  {\bf 110}, 1-7.

\bibitem[(2017)]{bar2} Barmalzan, G., Najafabadi, A.T.P., and Balakrishnan, N. (2017). Ordering properties of the smallest and largest claim amounts in a general scale model. {\it Scandinavian Actuarial Journal},  {\bf 2017(2)}, 105-124.

%\bibitem[(2017)]{bou} Bourguignon, M., Le\~{a}o, J., Leiva, V., and Santos-Neto, M. (2017). The transmuted Birnbaum-Saunders distribution. {\it REVSTAT Statistical Journal},  {\bf 15}, 601-628.
\bibitem[(2018)]{bar2018} Barmalzan, G., Najafabadi, A.T.P. and Balakrishnan, N. (2018). Some new results on aggregate claim amounts from two heterogeneous Marshall-Olkin extended exponential portfolios. {\it Communications in Statistics-Theory and Methods},  {\bf 47(11)}, 2779-2794.

%\bibitem[(2015)]{cw} Cai, J., and Wei, W. (2015). Notions of multivariate dependence and their applications in optimal portfolio selections with dependent risks. {\it Journal of Multivariate Analysis},  {\bf 138}, 156-169.
%\bibitem{flld} Fang, R., Li, C. and Li, X. (2016). Stochastic comparisons on sample extremes of dependent and heterogeneous observations. {\it Statistics}, {\bf 50(4)} 930-955.

\bibitem[(1978)]{clay} Clayton, D.G. (1978). A model for association in bivariate life tables and its application in epidemiological studies of familial tendency in chronic disease incidence. {\it Biometrika},  {\bf 65}, 141-151.

\bibitem[(1992)]{cox} Cox, D.R.(1972). Regression Models and Life Tables. {\it Journal of the Royal Statistical
Society, Series B},  {\bf 34(2)}, 187-220.

\bibitem[(2006)]{defr} Denuit, M., and Frostig, E. (2006). Heterogeneity and the need for capital in the individual model. {\it Scandinavian Actuarial Journal}, {\bf 2006(1)}, 42-66.
%
%\bibitem[(2015)]{fz15} Fang, L. and Zhang, X. (2015). Stochastic comparisons of parallel systems with exponentiated Weibull components. {\it Statistics and Probability Letters}, {\bf 97}, 25-31.
%
\bibitem[(2014)]{dd} Dolati, A., and Dehgan Nezhad, A. (2014). Some results on convexity and concavity of multivariate copulas. {\it Iranian Journal of Mathematical Sciences and Informatics}, {\bf 9(2)}, 87-100.

\bibitem[(2006)]{dur} Durante, F. (2006). {\it New results on copulas and related concepts}. ${\rm Universit\grave{a}}$ degli Studi di Lecce.

%\bibitem[(1936)]{eyr} Eyraud, H. (1936). Les principes de la mesure des correlations. {\it Ann. Univ. Lyon, III. Ser., Sect. A}, {\bf 1}, 30-47.

%\bibitem[(1960)]{far} Farlie, D.J. (1960). The performance of some correlation coefficients for a general bivariate distribution. {\it Biometrika}, {\bf 47(3-4)}, 307-323.

\bibitem[(2008)]{fin} Finkelstein, M. (2008). {\it Failure rate modeling for reliability and risk}. Springer, London.

\bibitem[(1979)]{frank} Frank, M.J. (1979). On the simultaneous associativity of $F(x,y)$ and $x+y-F(x,y)$. {\it Aequationes mathematicae}, {\bf 19(1)}, 194-226.

\bibitem[(2001)]{fro} Frostig, E. (2001). A comparison between homogeneous and heterogeneous portfolios. {\it Insurance: Mathematics and Economics}, {\bf 29(1)}, 59-71.

%\bibitem[(2015)]{gup} Gupta, N., Patra, L.K. and Kumar, S. (2015). Stochastic comparisons in systems with Fr$\grave{\rm e}$chet distributed components. {\it Operations Research Letters}, {\bf 43(6)}, 612-615.
%
%\bibitem{gk} Gupta, R.D., Kundu, D., 1999. Generalized exponential distributions. Aust. Nz. J. Statist. 41, 173-188.

%\bibitem[(2015)]{gr} Granzotto, D. C. T., and Louzada, F. (2015). The transmuted log-logistic distribution: Modelling, inference, and an application to a polled tabapua race time up to first calving data.  {\it Communications in Statistics-Theory and Methods}, {\bf 44(16)}, 3387-3402.

%\bibitem[(1960a)]{gum} Gumbel, E.J. (1960). Bivariate exponential distributions. {\it Journal of the American Statistical Association}, {\bf 55(292)}, 698-707.

\bibitem[(1960)]{gum2} Gumbel, E.J. (1960). Distributions des valeurs extremes en plusiers dimensions. {\it Publications de lInstitut de statistique de lUniversit\'e de Paris}, {\bf 9}, 171-173.

\bibitem[(1998)]{gup} Gupta, R.C., Gupta, P.L., and Gupta, R.D. (1998). Modeling failura time data by Lehman alternatives. {\it Communications in Statistics-Theory and Methods}, {\bf 27(4)}, 887-904.

\bibitem[(2017)]{hami} Hami Golzar, N., Ganji, M., and Bevrani, H. (2017). The Lomax-exponential distribution, some properties and application. {\it Journal of Statistical Research of Iran}, {\bf 13(2)}, 131-153.

\bibitem[(1948)]{harris} Harris, T.E. (1948). Branching processes. {\it The Annals of Mathematical Statistics}, {\bf 19}, 474-494.

\bibitem[(2004)]{huru} Hu, T., and Ruan, L. (2004). A note on multivariate stochastic comparisons of Bernoulli random variables. {\it Journal of Statistical Planning and Inference}, {\bf 126(1)}, 281-288.

% \bibitem[(2014)]{iras} Iriarte, Y. A., and Astorga, J. M. (2014). Transmuted Maxwell probability distribution. {\it Revista Integraci\'{o}n}, {\bf 32(2)}, 211-221.

\bibitem[(1963)]{kar} Karlin, S., and Novikoff, A. (1963). Generalized convex inequalities. {\it Pacific Journal of Mathematics}, {\bf 13(4)}, 1251-1279.

%\bibitem[(2017)]{keyi} Kemaloglu, S. A., and Yilmaz, M. (2017). Transmuted two-parameter Lindley distribution. {\it Communications in Statistics-Theory and Methods}, {\bf 46(23)}, 11866-11879.

\bibitem[(2008)]{khah} Khaledi, B.E., and Ahmadi, S.S. (2008). On stochastic comparison between aggregate claim amounts. {\it Journal of Statistical Planning and Inference}, {\bf 138(7)}, 3121-3129.


%\bibitem[(1974)]{kim} Kimberling, C.H. (1974). A probabilistic interpretation of complete monotonicity. {\it Aequationes mathematicae}, {\bf 10(2-3)}, 152-164.
%\bibitem[(2017)]{khan} Khan, M.S., King, R., and Hudson, I.L. (2017). Transmuted Weibull distribution: Properties and estimation.  {\it Communications in Statistics-Theory and Methods}, {\bf 46(11)}, 5394-5418.

%\bibitem[(1974)]{kim} Kimberling, C. H. (1974). A probabilistic interpretation of complete monotonicity. {\it %Aequationes mathematicae}, {\bf 10}, 152-164.

%\bibitem[(2016)]{kc16} Kundu, A. and Chowdhury, S. (2016). Ordering properties of order statistics from heterogeneous exponentiated Weibull models. {\it Statistics and Probability Letters}, {\bf 114}, 119-127.
%
%\bibitem{lem} Lemonte, A. J., 2013. A new exponential-type distribution with constant, decreasing, increasing, upside-down bathtub and bathtub-shaped failure rate function. Comput. statist. data anal. 62, 149-170.
%\bibitem{lfr}
%Li, X. and Fang,  R., 2015. Ordering properties of order statistics from random variables of Archimedean copulas with applications, J. Multivariate Anal. 133 304-320.
%
%\bibitem[(2015)]{ll} Li, C. and  Li, X. (2015). Likelihood ratio order of sample minimum from heterogeneous Weibull random variables. {\it Statistics and Probability Letters}, {\bf 97}, 46-53.
%
%\bibitem[(1966)]{leh} Lehmann, E.L. (1966). Some concepts of dependence. {\it The Annals of Mathematical Statistics}, {\bf 37}, 1137-1153.

\bibitem[(1994)]{kukl} Kumar, D., and Klefsj\"{o}, B. (1994). Proportional hazards model: a review. {\it Reliability Engineering and System Safety}, {\bf 44(2)}, 177-188.

\bibitem[(2016)]{lili} Li, C., and Li, X. (2016). Sufficient conditions for ordering aggregate heterogeneous random claim amounts. {\it Insurance: Mathematics and Economics}, {\bf 70}, 406-413.

\bibitem[(2018)]{lili2} Li, C., and Li, X. (2018). Stochastic comparisons of parallel and series systems of dependent components equipped with starting devices. {\it Communications in Statistics-Theory and Methods}, DOI: 10.1080/03610926.2018.1435806.


\bibitem[(2013)]{lll} Li, H., and Li, X. (2013). {\it Stochastic Orders in Reliability and Risk}. Springer, New York.
%

\bibitem[(2000)]{ma} Ma, C. (2000). Convex orders for linear combinations of random variables. {\it Journal of Statistical Planning and Inference}, {\bf 84}, 11-25.

\bibitem[(1997)]{marolk} Marshall, A.W., and Olkin, I. (1997). A new method for adding a parameter to a family of distributions with application to the exponential and Weibull families. {\it Biometrika}, {\bf 84(3)}, 641-652.

\bibitem[(2011)]{met} Marshall, A.W., Olkin, I., and Arnold, B.C. (2011). {\it Inequalities: Theory of Majorization and its Applications}. Springer, New York.

%\bibitem[(2009)]{ncn} McNeil, A. J., and Ne$ \breve{s} $lehov$ \acute{a} $, J. (2009). Multivariate Archimedean Copulas, d-Monotone Functions and $ \ell_{1} $-Norm Symmetric Distributions. The Annals of Statistics, 3059-3097.

%\bibitem[(2008)]{mir2} Mirhossaini, S.M., and Dolati, A. (2008). On a new generalization of the exponential distribution. {\it Journal of Mathematical Extension}, {\bf 3(1)}, 27-42.

%\bibitem[(2011)]{mir1} Mirhossaini, S.M., Dolati, A., and Amini, M. (2011). On a class of distributions generated by stochastic mixture of the extreme order statistics of a sample of size two. {\it Journal of Statistical Theory and Application}, {\bf 10}, 455-468.

%\bibitem[(1956)]{morgen} Morgenstern, D. (1956). Einfache beispiele zweidimensionaler verteilungen. {\it Mitteilingsblatt fur Mathematische Statistik}, {\bf 8}, 234-235.

%\bibitem{ncn} McNeil,  A. J. and Ne$ \breve{s} $lehov$ \acute{a} $, J. (2009). Multivariate Archimedean Copulas, d-Monotone Functions and $ \ell_{1} $-Norm Symmetric Distributions. Ann. Statist. 3059-3097.

%\bibitem{ms}Mudholkar, G.S., Srivastava, D.K., 1993. Exponentiated Weibull family for analyzing bathtub failure-rate data. IEEE Trans. Reliab. 42, 299–302.

%\bibitem{nh} Nadarajah, S., Haghighi, F., 2011. An extension of the exponential distribution. Statistics 45, 543-558.

%\bibitem{nels} Nelsen, R.B., 2006. An introduction to copulas. Springer, New York.

\bibitem[(2002)]{must} M\"{u}ller, A., and Stoyan, D. (2002). {\it Comparison methods for stochastic models and risks}. John Wiley \& Sons, New York.

\bibitem[(2018)]{nato} Nadeb, H., and Torabi, H. (2018). Stochastic comparisons of series systems with independent heterogeneous Lomax-exponential components. {\it Journal of Statistical Theory and Practice}, {\bf 12(4)}, 794-812.

\bibitem[(2007)]{nel} Nelsen, R.B. (2007). {\it An Introduction to Copulas}. Springer Science \& Business Media, New York.

%\bibitem[(2016)]{ok} Okorie, I. E., Akpanta, A. C., and Ohakwe, J. (2016). Transmuted Erlang-truncated exponential distribution.  {\it Economic Quality Control}, {\bf 31(2)}, 71-84.

%\bibitem[(2016)]{sab} Saboor, A., Elbatal, I., and Cordeiro, G. M. (2016). The transmuted exponentiated Weibull geometric distribution: Theory and applications.  {\it Hacettepe Journal of Mathematics and Statistics}, {\bf 45}, 973-987.

%\bibitem[(2006)]{nel} Nelsen R.B. (2006). An introduction to copulas. Springer, New York.

\bibitem[(2007)]{ss} Shaked, M., and Shanthikumar, J.G. (2007). {\it Stochastic Orders}. Springer, New York.

%\bibitem[(2015)]{to15} Torrado, N. (2015). Comparisons of smallest order statistics from Weibull distributions with different scale and shape parameters. {\it Journal of the Korean Statistical Society}, {\bf 44(1)}, 68-76.

%\bibitem[(2015)]{tk15} Torrado, N. and Kochar, S. C. (2015). Stochastic order relations among parallel systems from Weibull distributions. {\it Journal of Applied Probability}, {\bf 52(1)}, 102-116.

%\bibitem[(2009)]{shbu} Shaw, W.T., and Buckley, I.R. (2009). The alchemy of probability distributions: beyond Gram-Charlier expansions, and a skew-kurtotic-normal distribution from a rank transmutation map. {\it arXiv preprint arXiv:0901.0434}.

%\bibitem[(2014)]{tian} Tian, Y., Tian, M., and Zhu, Q. (2014). Transmuted linear exponential distribution: A new generalization of the linear exponential distribution. {\it Communications in Statistics-Simulation and Computation}, {\bf 43(10)}, 2661-2677.

\bibitem[(2015)]{zz} Zhang, Y., and Zhao, P. (2015). Comparisons on aggregate risks from two sets of heterogeneous portfolios. {\it Insurance: Mathematics and Economics}, {\bf 65}, 124-135.
\end{thebibliography}
%%% before submission, comment out the above two lines and
%%% cut and paste your .bbl file here:
%\end{linenumbers}
\end{document}